\documentclass[12pt]{article} 
\usepackage{amsfonts, amssymb, amsmath, amsthm, tikz-cd, xcolor, mathrsfs} 
\usepackage[mathscr]{eucal}
\usepackage[utf8]{inputenc}

\swapnumbers
\newtheorem{theorem}{Theorem}[section] 
\newtheorem{lemma}[theorem]{Lemma} 
\newtheorem{corollary}[theorem]{Corollary} 
\newtheorem{proposition}[theorem]{Proposition} 

\theoremstyle{definition} 
\newtheorem{definition}[theorem]{Definition} 
\newtheorem{remark}[theorem]{Remark} 
\newtheorem{remarks}[theorem]{Remarks}
\newtheorem{example}[theorem]{Example} 
\newtheorem{examples}[theorem]{Examples}

\newcommand{\blue}{\color{blue}}

\newcommand{\1}{{\mathbf 1}}
\newcommand{\A}{{\mathcal A}}
\newcommand{\Aff}{\mbox{Aff}}
\newcommand{\B}{{\mathcal B}}
\newcommand{\co}{~\mbox{co}~}
\newcommand{\Cat}{{\mathcal C}}
\newcommand{\D}{{\mathcal D}}
\newcommand{\Dacey}{{\mathbb D}}
\newcommand{\E}{{\boldsymbol A}}

\newcommand{\Ev}{{\mathcal E}}
\newcommand{\F}{{\boldsymbol B}}

\newcommand{\G}{\boldsymbol{G}}
\newcommand{\Graph}{\mbox{Graph}}
\renewcommand{\H}{\mathscr{H}}
\renewcommand{\hat}{\widehat}

\renewcommand{\iff}{\mbox{ iff }}
\newcommand{\J}{\boldsymbol{\mathscr J}}
\renewcommand{\k}{{\mathbb K}}
\newcommand{\K}{\mathscr{K}}
\newcommand{\M}{{\mathcal M}}
\newcommand{\Markov}{\mbox{\bf Markov}}
\newcommand{\mintensor}{\otimes_{\mbox{\small min}}}
\newcommand{\maxtensor}{\otimes_{\mbox{\small max}}}
\newcommand{\Mod}{\mbox{Mod}}
\newcommand{\N}{{\mathbb N}}
\newcommand{\Ntest}{{\mathcal N}}
\newcommand{\Ob}{{\mathcal O}}
\newcommand{\op}{\mbox{\small op}}
\newcommand{\OUS}{\mbox{\bf OUS}}
\newcommand{\Pow}{{\mathcal P}}
\newcommand{\Prob}{\mbox{\bf Prob}}

\newcommand{\R}{{\mathbb R}}
\newcommand{\ran}{\mbox{ran}}
\newcommand{\sa}{\mbox{sa}}
\newcommand{\spn}{\mbox{span}}
\newcommand{\Test}{\mbox{\bf Test}}

\newcommand{\tempout}[1]{{}}
\renewcommand{\tilde}{\widetilde}
\renewcommand{\u}{u} 
\newcommand{\V}{{\mathbf V}}

\title{\sf Order-unit spaces and probabilistic models} 
\author{\sf John Harding\footnote{Department of Mathematics and Statistics, New Mexico State University} \footnote{JH acknowledges the support of NSF grant DMS-2231414.} and Alex Wilce\footnote{Department of Mathematics and Computer Science, Susquehanna University}}
\date{}
\begin{document} 

\maketitle

\begin{abstract}  We exhibit a functor from the category of order-unit spaces and positive, unit-preserving mappings into the category $\Prob$ of probabilistic models (test spaces with designated state spaces) and morphisms thereof. Restricted to any subcategory of OUS monoidal with respect to a bilinear composition rule, our functor is also monoidal. This much shows that the convex-operational approach to physical theories can be subsumed by the test-space approach, without resort to ``generalized test spaces''.  A second, less functorial, construction, equipping a probabilistic model with tests representing ``weighted coins'', also sheds light on the nature of unsharp observables. 

\end{abstract}

\section{Introduction} 
The novel probabilistic features of quantum mechanics have inspired attempts to generalize classical probability theory so as to allow for random quantities that are not jointly measurable. Within such a framework, one then hopes to characterize (or {\em reconstruct}) quantum mechanics in some probabilistically meaningful way.  Given this aim, it is important that the probabilistic framework not presuppose any linear or $\ast$-algebraic structure not already implicit in basic probabilistic ideas. 

Broadly speaking, there have been two approaches to constructing such a general framework. One, going back to Mackey \cite{Mackey} and further 
developed by Foulis and Randall \cite{FREL, FR-Dirac, FGR-II}, begins with a {\em test space}: a catalogue $\M$ of discrete, classical experiments or measurements (or rather, their outcome-sets); states are defined in terms of consistent assignments of probabilities to the outcomes of these, and the rest of the theory is built up on this basis. 
The other, {\em convex-operational} approach \cite{Davies-Lewis, Edwards, Ludwig, Mielnik}, begins with a set $\Omega$ of possible 
states --- a {\em state-space} --- which is assumed to be convex, to reflect the possibility of randomizing the preparation of states. 
Dually,  one can begin with an OUS $(\E,u)$, representing states by positive affine functionals $\alpha$  on $\E$ with $u(\alpha) = 1$, and 
taking elements $a \in \E$ with $0 \leq a \leq u$ --- usually 
termed {\em effects} --- to represent measurement outcomes. An 
experiment is represented by a {\em discrete observable}, 
which is a sequence $(a_1,...,a_n)$ of non-zero effects summing to $u$.


The purpose of this note is to clarify and formalize the connection between these two broad approaches. In one direction, this is straightforward. If $\Omega$ is a convex set of probability weights 
on a test space $\M$, then $\Omega$ generates a cone-base space $\V(A)$, which in many cases --- in particular, whenever $\Omega$ is pointwise 
compact --- is a complete base-normed space \cite{Wilce-GPT}, and every test in $\M(A)$ defines a discrete observable with values in $\V(A)^{\ast}$.  We can set up a category, $\Prob_{c}$, of probabilistic models having compact, convex state spaces, and then this construction defines a faithful functor $\Prob_{c} \rightarrow \OUS$.  

The route in the other direction seems not to have been described in the literature. 
The usual  analogue of a ``test" in the convex-operational approach is a discrete observable, i.e., a mapping $a : I \rightarrow (0,u]$ with $\sum_{i\in I} a_i = u$, where $I$ is a finite set. Where $I = \{1,...,n\}$, this is simply a {\em list} $(a_1,...,a_n)$ of effects $a_i \in \E$ with $\sum_{i=1}^{n} a_i = u$. Such a list may contain repetitions --- indeed, $(\tfrac{1}{n} u,...,\tfrac{1}{n}u)$ is a perfectly good $\{1,...,n\}$-valued observable. This has led some authors \cite{Pulmannova-Wilce, Gudder} to consider various kinds of generalized test spaces in which outcomes are permitted to have ``multiplicities" greater than $1$. However, as we shall see, no such generalization is necessary. The key observation is that, given a discrete observable $a : I \rightarrow \E$, its {\em graph} ---  the set of pairs $(i,a_i)$ with $i \in I$ --- is the appropriate model of the corresponding experiment or test. 

Accordingly, we construct, for every OUS $\E$ and every collection $\J$ of index sets, a probabilistic model $\Mod^{\J}(\E)$ consisting of graphs of $J$-valued observables, $J \in \J$ and show that this provides a functor from $\OUS$ to $\Prob_{c}$. When $\J$ is reasonably large, this is faithful. If $\J$ is closed under Cartesian products and $\Cat$ is any sub-category of $\OUS$ equipped with a bilinear composition rule making $\Cat$ a symmetric monoidal category, $\M^{\J}$, as restricted to $\Cat$, is a monoidal functor.  In this sense, the convex-operational framework is a special case of that based on tests.  



This paper is organized as follows. Section 2 develops some useful background information on vector-valued weights on test spaces. Section 3 introduces the $\M^{\J}$ construction, and explores some of its basic properties, while section 4 establishes that $\M^{\J}$ has the claimed functorial properties. Appendix A contains the proof of a minor but necessary result about order unit spaces.  Appendix B discusses how these ideas play out in the classical case in which $\E$ is the space of bounded measurable functions on some measurable space. 
Appendix C contains details concerning the monoidality of the $(~\cdot~)^{\J}$ functor. This includes a brief review of non-signalling composites of probabilistic models. Appendix D discusses an alternative, albeit less functorial, construction that provides some further operational insight into ``unsharp'' observables.  



\section{$\E$-valued weights on Test Spaces} 

It will be useful to start with a general discussion of ``probability weights'' on a test space taking values in an order-unit space.  
We assume from now on that the reader is more or less familiar with basic terminology regarding order-unit spaces (see, e.g., \cite{Alfsen} or \cite{Alfsen-Shultz}) and with test spaces, probabilistic models (test spaces 
with distinguished state spaces), and their morphisms (see, e.g., \cite{Wilce-GPT}).  We recall that any compact convex set $K$ generates 
a canonical base-normed Banach space $\V(K)$; the dual of this 
is an order unit space isomorphic to the space $\Aff(K)$ of 
bounded affine functionals on $K$. The space $\Aff_{c}(K)$ of 
continuous affine functionals on $K$ is also an OUS, with 
$\V(K) \simeq \Aff_{c}(K)^{\ast}$. In particular, $K$ is 
isomorphic to the state space of $\Aff_{c}(K)$ \cite{Alfsen}. 

We also recall that a {\em probability weight} on a test space 
$\M$ with outcome-space $X = \bigcup \M$ is a function 
$\alpha : X \rightarrow \R$ with $0 \leq \alpha(x)$ for all $x \in X$ 
and $\sum_{x \in E} \alpha(x) = 1$ for all tests $E \in \M$. 
A {\em probabilistic model} is a pair $(\M,\Omega)$ consisting 
of a test space and a distinguished set $\Omega$ of probability 
weights. In this paper, we will always assume that $\Omega$ 
is convex and compact in $\R^{X}$. 


In what follows, $\M$ is a test space with outcome-set $X$ and $\E$ is an order-unit space with order-unit $u$, positive cone $\E_+$, state-space $S(\E)$, and dual base-normed space $\E^{\ast} =: \V$. We write $[0,u]$ for the set of {\em effects}, elements $a \in \E$  with $0 \leq a \leq u$. 

\begin{definition} A {\em normalized $\E$-valued weight} on $\M$ is a mapping $F : X \rightarrow \E_{+}$ satisfying 
\[\sum_{x \in E} F(x) = u\]
for every test $E \in \M$. 
\footnote{Where $E$ is infinite, the sum can be understood as a norm limit, or as a weak limit of finite
partial sums. Since every term $F(x)$ is positive and bounded by $u$, these coincide, and the sum is simply the supremum of these partial sums.} 
\end{definition}.

In what follows, we will be concerned almost exclusively with normalized weights, and therefore drop the adjective ``normalized'', taking it as understood. Such weights could also be called {\em representations} of $\M$ on (or in) $\E$, or possibly even {\em $\M$-valued observable} on $\E$, for reasons spelled out below (see Examples (1) and (2). 

\begin{examples}
(1) If $\E = \B_{sa}(\H)$ for a Hilbert space $\H$, with the usual ordering and order-unit $\1$, an $\E$-valued weight could also be described as a ``positive-operator  valued weight'' on $\M$. In the special case in which $\M$ is the space of countable partitions of a measurable space $(S,\Sigma)$ by $\Sigma$-measurable  sets, we recover the usual definition of a POVM on $\H$.   In this case, it would also be standard to say that $F$ is an $S$-valued observable on $\H$, suggesting  the alternative terminology mentioned above. 


(3) More generally, suppose $\M$ is the test space $\M(\B)$ of countable partitions of the unit in any Boolean algebra $\B$, with $X = \B \setminus \{0\}$. Then an $\E$-valued weight $F : X \rightarrow \E$ extends uniquely to a positive, $\E$-valued measure on $\B$ --- also denoted by $F$ --- simply by setting $F(0_{\B}) = 0 \in \E$. Conversely, every normalized, positive $\E$-valued measure on $\B$ defines an $\E$-valued weight on $\B$. We will refer to a state of this kind as {\em classical} observable on $\E$. 
 
(3) If $\V = \V(\Omega)$ for some set $\Omega$ of probability weights, on $\M$, then  every outcome of $X$ gives rise to an evaluation-functional defined by $\hat{x} : \alpha \mapsto \alpha(x)$ for all $\alpha \in \V$, and $F : x \mapsto \hat{x}$ is a  $\V^{\ast}$-valued weight on $\M$. We call this the {\em canonical} $\V^{\ast}$-valued weight  associated with the model $A = (\M,\Omega)$.  

(4) Let $K$ be a compact convex set, with corresponding complete base-normed space $\V(K)$.  For any probabilistic model $(\M,\Omega)$ and affine mapping $\phi : K \rightarrow \Omega$, we have a $\V(K)^{\ast}$-valued state on $\M$ given by $F = \phi^{\ast} : \bigcup \M \rightarrow \Aff(K) \simeq \V(K)^{\ast}$ given by $F(x)(k) = \phi(k)(x)$.

(5) If $(\E,u)$ is any OUS, let $\D(\E)$ be the collection of 
{\em decompositions of the unit}, i.e., finite sets $E \subseteq (0,u]$ 
with $\sum E = u$. This is a test space, with outcome-set 
$\bigcup \D(\E) = (0,u]$; every state on $\E$ 
restricted to $(0,u]$, is a probability weight $\D(\E)$, and one can show that every probability weight on $\D(\E)$ extends to a state on $\E$ (\cite{Wilce-GPT}). Evidently, the inclusion mapping $(0,u]$ is an $\E$-valued weight on $\D(\E)$. 

\end{examples} 

Returning to the general case, let $F$ be an $\E$-valued weight on a test space $\M$, and suppose $\phi \in S(\E)$.   Then $F^{\ast}(\phi) := \phi \circ F$ is a (real-valued) probability weight on $\M$. Thus, we have a probabilistic model $A = (\M,F^{\ast}(\Omega))$. The mapping $\phi \mapsto F^{\ast}(\phi)$ extends to a positive linear mapping $F^{\ast} : \E^{\ast} \rightarrow \V(\Omega)$, and this in turn gives us a positive linear mapping $F^{\ast \ast} : \V(\Omega)^{\ast} \mapsto \E^{\ast \ast}$. Letting $F' : x \mapsto \hat{x}$ be the canonical $\V(A)^{\ast}$-valued state defined in Example (3) above, we now have 
\[F^{\ast \ast}(F'(x))(\alpha) = F'(x)(F^{\ast}(\alpha)) = F^{\ast}(\alpha)(x) = \alpha(F(x))\]
for every $x \in X$. Since $S(\E)$ separates points of $\E$, it follows that 
\[F^{\ast \ast} \circ F' = F.\]
In other words,  every $\E$-valued weight on $\M$ factors into (i) the canonical $\V(\Omega)$-valued weight $F'$ of a model $A = (\M,\Omega)$, followed by (ii) a positive, normalized linear mapping $\V^{\ast}(\Omega) \rightarrow \E$. 

\begin{remarks} 
(1) In general, the set $F^{\ast}(S(\E))$ of probability weights on $\M$ arising from a given weight $F$ will not be separating for outcomes of $\M$. An extreme case of this occurs when, for some $x \not = y$ in $X$, $F(x) = F(y)$. 

(2) We can define an $\E$-valued weight on a probabilistic model $(\M,\Omega)$ to be an $\E$-valued weight $F$ on $\M$ such that $F^{\ast}(S(\E)) \subseteq \Omega$. This is very close to the definition of a morphism of models, except that $\E$ is not a probabilistic model. It does not help to replace $\E$ with the model $(\D(\E),S(\E))$ where 
$\D(E)$ is the test space of unordered partitions of unity 
in $\E$: where $F(x) = F(y)$ for some  pair of outcomes with 
$x \perp y$, $F$ is not locally injective, and thus not a morphism of test spaces. 
\end{remarks}

\tempout{
{\bf 1.4 Some Questions:} 

(i) Can we characterize the space $\Pr(\M(\B,\E))$? 

(ii) Let $\Delta_{\E}(\B)$ denote the set 
of finitely-additive probability measures on $\B$ having the form $F^{\ast}(f)$ where $f$ is a state on $\E$ and $F$ is 
an $\E$-valued weight on $\B$. Note that for $\E = \R$, 
$\Delta_{\E}(\B)$ is just the set of all f.a. probability measures on $\B$.  Is this still true if $\E$ is any 
f.d. OUS? What about for an arbitrary OUS?

(iii) If we fix an observable $F \in \Ob(\B,\E)$, 
we can consider the set $\M(F) \subseteq \M(\B,\E)$ consisting 
of all tests arising from this fixed $F$. That is, 
$\M(F)$ consists of sets $\{F(b_1),...,(b_n)\}$ where $\{b_1,...,b_n\}$ is a partition of unity in $\B$. Is 
$\M(F)$ a refinement ideal? 

(iv) When, if ever, is $\M(\B,\E)$ algebraic?\\
}

\section{Test spaces associated with an OUS}

A very natural way of attaching a test space to an OUS $(\E,u)$ is the following: let $\D(A)$ denote the set of all decompositions of the unit, that is, finite 
subsets $E \subseteq (0,u]$ summing to $u$. This is a test space, with outcome-set $\bigcup \D(\E) = (0,u]$.  Any state 
on $\E$, restricted to $(0,u]$, defines a probability weight 
on $\D(\E)$, and one can show \cite{Wilce-GPT} that conversely, 
every probability weight on $\D(\E)$ extends uniquely to a state  on $\E$.  If $(a_1,...,a_n)$ is an observable with distinct  terms, then $\{a_1,...,a_n\}$ belongs to $\D(\E)$; 
however, $\D(\E)$ does not capture observables involving repeated effects. A related point is that if $\Phi : \E \rightarrow \F$  is a unit-preserving positive mapping, then unless  $\Phi$ is injective, it does not give rise to a morphism  $\D(\E) \rightarrow \D(\F)$.

Let $I$ be a non-empty set. An $I$-valued {\em observable} on an OUS $\E$ is simply an $I$-indexed partition of unity, that is, a function $f : I \rightarrow (0,u]$ with $\sum_{x \in I} f(x) = u$.  (Note here that we do not allow $f(x) = 0$.) If $I = \{1,...,n\}$, we recover the idea of a discrete observable as a list of effects summing to $u$. We will refer to such a list as an {\em ordered} observable. 

If $f : I \rightarrow E$ is an $I$-valued observable, let $G(f)$ denote its graph, i.e., the set $\{(x,a) | x \in I, f(x) = a > 0\}$. We propose to consider this graph as the outcome-set for the experiment of actually measuring $f$. Let $\Ob(I,\E)$ denote the set of all graphs of $I$-indexed observables on $\E$. It  is not hard to see that $\Ob(I,\E)$ is irredundant, and therefore, a test space, with outcome-set $\bigcup\Ob(I,\E) = I \times (0,u]$. 

Operationally, for a given $f \in \Ob(I,\E)$, the effect $f(x) = a$ is a mathematical model of the physical event, {\em relating to the target system}, corresponding to the directly observable experimental outcome $x$. Distinct pairs $(x,a)$ and $(y,a)$, where $x, y \in I$, can arise when two distinct experimental outcomes can be triggered by the same event in the system. A distinct pair $(x,a)$ and $(x,b)$ can arise as outcomes of distinct experiments that happen to ``label" distinct physical events in the same way. 

The structure of $\Ob(I,\E)$ is quite large, and combinatorially very rich, even if $I$ is small and $E$ has low dimension. As an illustration, let $I = \{1,2,3\}$ and let $\E = [0,1]$. Then $\Ob(I,\E)$ consists of graphs of triples $(a,b,c)$ of positive real numbers summing to $1$. Geometrically, these are the points in the interior 
of the standard simplex in $\R^{3}$. Each such triple 
corresponds to a three-outcome test, where we identify $(a,b,c)$ with the set $\{(1,a), (2,b), (3,c)\}$. If $a,b,c$ are distinct, any permutation of the entries gives a distinct test; however, these overlap in a complex way. Fixing attention on a single triple of distinct constant $a$, $b$, $c$, one finds that the sub-test space consisting of graphs of the six triples corresponding to permutations of $(a,b,c)$ can be organized as the rows and columns of a $3 \times 3$ grid \cite{BW-Foils}.  The test space $\Ob(I,E)$ will contain a copy of this six-test 
arrangement for every distinct triple $(a,b,c)$. Moreover, these will overlap; e.g., $(a,u,v)$ and $(a,b,c)$ intersect at $(1,a)$. 

The basic idea now is to bundle together all $I$-indexed observables, for arbitrary (say) finite sets $I$. Of course, we can't do quite this, as the latter form a proper class. But nothing prevents us from considering large collections of finite index sets. 

\begin{definition} If $\J$ is any set of finite sets, let 
\[\M^{\J}(\E) \ =  \ \bigcup_{J \in \J} \Ob(J,\E).\] 
When $\J = \{\{1,...,n\} \mid n \in \N\}$, $\M^{\J}$ consists of the graphs of all ordered observables on $\E$. In this case, we write $\M^{o}(\E)$ for $\M^{\J}(\E)$. 
\end{definition} 
Every state $f$ on $\E$ gives rise to a probability weight on $\M^{\J}(\E)$, given by $\alpha_{f}(j,a) = f(a)$ for all $j \in \bigcup \J$ and $a \in (0,a]$. Letting $\Omega^{\J}(\E) = \{ \alpha_f \mid f \in S(\E)\}$, we have a probabilistic model associated with any OUS $\E$ and any collection $\J$, namely \[\Mod^{\J}(\E) = (\M^{\J}(\E), \Omega^{\J}(\E)).\] 
In order to reduce notation, whenever possible 
we henceforth identify $\Omega^{\J}(\E)$ with $S(\E)$.

In many cases, $\Omega^{\J}(\E) = \Pr(\M^{\J}(\E))$. That is, all probability weights on $\M^{\J}(\E)$ have the form $\alpha_f$ for a unique state on $E$. For the following, let $\Ev(\J)$ denote the set of all subsets of index sets in $\J$; that is, $\Ev(\J) = \bigcup_{J \in \J} \Pow(J)$. If $\M$ is any test space, elements of $\Ev(\M)$ are 
called {\em events} of $\M$. Events $a, b \in \Ev(\M)$ are 
{\em complementary}, written $a \co b$, iff $a \cap b = \emptyset$ and $a \cup b \in \M$. Thus, $\emptyset \co E$ for 
any test $E \in \M$, and non-empty events are complements, or complementary, iff they partition some test. When two events $a$ and $b$ share a complement, they are said to be {\em perspective}, written 
$a \sim b$. It is easy to see that in this case $\alpha(a) = \alpha(b)$ for all probability weights $\alpha$ on $\M$.

\begin{lemma} Suppose that for all $i, j \in \bigcup \J$, there exists some $r \in \bigcup \J$ such that $\{i,r\}, \{j,r\} \in \Ev(\J)$. Then every probability weight on $\M^{\J}(\E)$ has the form $\tilde{\alpha}$ for a state on $\E$.
\end{lemma} 

{\em Proof:} Let $(i,a), (j,a)$ be two distinct outcomes of $\M^{\J}(\E)$, and let $\beta$ be a state on $\M^{\J}(\E)$. We wish to show, first, that $\beta(i,a) = \beta(j,a)$. If $a = u$, then $\{(i,a)\}$ and $\{(j,a)\}$ are both tests, and so $\beta((i,a)) = \beta((j,a)) = 1$. 
If $a < u$, then $a' = u - a \not = 0$. Construct two observables $f_i$ and $f_j$ by selecting some $k \not \in \{i,j\}$ and defining $f_{i}(i) = a, f_{i}(r) = u - a$ and similarly for $f_j$. Then $\{(i,a), (r,a')\}$ and $\{(j,a), (r,a')\} \in \M^{\J}(\E)$, and so $\{(i,a)\} \sim \{(j,a)\}$, so again, $\beta((i,a)) = \beta((j,a))$.  We now see that the function  $\alpha : (0,u] \rightarrow [0,1]$ given by $\alpha(a) = \beta(i,\alpha)$ is well-defined, and $\beta  = \tilde{\alpha}$. It remains to show that $\alpha$ extends uniquely to a state on $\E$. This is a more or less standard result. (It is proved by \cite{Busch} in the case where $\E$ is the space of self-adjoint operators on a Hilbert space, but his proof applies verbatim to the general setting of an OUS. See also   \cite[Appendix E]{WW-Gleason}, \cite{Wilce-GPT}.) $\Box$ 

The hypothesis of the Lemma is satisfied when $\J$ is the collection of all initial segments of $\N$. Thus, probability weights on $\M^{o}(\E)$ are (in effect) states on $\E$. 

{\bf Connections with quantum logic} 
In general, $\J$ is not irredundant, and thus, can't be construed as a test space. However, given a test space $\A$, we can certainly form $\M^{\A}(\E)$. 

{\em Notation:} If $A \in \Ev(\M^{\J}(\E))$, say with $A \subseteq G(f)$ for an observable $f : E \rightarrow \E$. Let 
\[A_o = \pi_1(A) = \{ x \in \bigcup \J | \exists a \in (0,\u] ~(x,a) \in A\}\]
and let $\sum A := \sum\{ a | (x,a) \in A\} = \sum_{x \in A_o} f(x)$.

\begin{lemma}\label{lem: complementary events}  
Let $\A$ be a test space, and let $A, B$ denote events of $\M^{\A}(\E)$.  
Then $A \co B$ iff $A_o \co B_o$ in $\Ev(\A)$ and $\sum B = u - \sum A$. 
\end{lemma}

{\em Proof:} If $A \co B$, then $A \cap B = \emptyset$, and $A \cup B = G(f)$ for some $f \in \Ob(E,\E)$ for some test $E \in \A$.   Thus, $\sum A + \sum B = \sum(A \cup B) = u$, so $\sum  B = u - \sum A$. Since $\pi_1$ is injective on $G(f)$, $A_o \cap B_o = \emptyset$ and $A_o \cup B_o = E$. Thus, $A_o \co B_o$ in $\Ev(\A)$.  For the converse, suppose 
$A \subseteq G(g)$ and $B \subseteq G(h)$ for some observables $g, h$ with domains in $\A$. If $A_o \co B_o$, with $A_o \cup B_o = E \in \M$, define $f : E \rightarrow \E$ by $f(x) = g(x)$ if $x \in A_o$ and $f(x) = h(x)$ if $x \in B_o$. Then $\sum_{x \in E} f(x) = \sum A + \sum B$. If the latter is $u$, then $f$ is an observable with domain in $\A$, hence, $G(f)$ is a test in $\M^{\A}(\E)$. 
We have $A \cap B = \emptyset$ (since $A_o \cap B_o = \emptyset$) and $A \cup B = G(f)$. Thus, $A \co B$ in $\Ev(\M^{\A}(\E))$. $\Box$ 

Recall that a test space $\A$ is {\em algebraic} iff perspective events in $\Ev(\A)$ share exactly the same complementary events. 

\begin{proposition} 
$\M^{\A}(\E)$ is algebraic iff $\A$ is algebraic. 
\end{proposition} 

{\em Proof:} Immediate from the Lemma. $\Box$ 

If $\A$ is algebraic, perspectivity is an equivalence relation on $\Ev(\A)$, and the equivalence classes can be organized in a natural way into an orthoalgebra $\Pi(\A)$, called the {\em logic} of $\A$. For details, see \cite{Wilce-TS}. The logic of $\M^{\A}(\E)$, where $\A$ is algebraic, is easy to characterize. If $L$ and $M$ are effect algebras, let 
\[L \ast M = \{ (a,b) \in L \times M | a = 0 \iff b = 0 \ , \ a = 1 \iff b = 1\}.\]
That is, 
\[L \ast M = L \times M \setminus ( \{0,1\} \times M \cup L \times \{0,1\}) \cup \{(0,0), (1,1)\}.\]
For $(a,b), (u,v) \in L\ast M$, define 
\[(a,b) \perp (u,v) \iff a \perp u, \ b \perp v, \ \mbox{and} \
(a \oplus u, b \oplus v) \in L \ast M.\]
In this case, let $(a,b) \oplus (u,v) = (a \oplus u, b \oplus v)$. It is easy to check that $(L\ast M, \oplus)$ is then an effect algebra 
with $0_{L\ast M} = (0,0)$ and $1_{L\ast M} = (1,1)$; moreover, the projection maps $L \ast M \rightarrow L$ and $L\ast M \rightarrow M$ are surjective effect-algebra homomorphisms, and if $L$ is an orthoalgebra, so is $L\ast M$.  

The proof of Lemma \ref{lem: complementary events} above gives us 

\begin{lemma} If $\A$ is algebraic, $\Pi(\M^{\A}(\E)) \simeq \Pi(\A) \ast [0,u]$.
\end{lemma} 

\begin{remark} As far as we know, the $L \ast M$ construction 
has not been discussed in the literature on OMLs and OMPs. 
It would be of interest to know what kinds of structure 
arise when one or both factors are Boolean, or $MV$ algebras. We leave this to future work.
\end{remark}

\section{Processes and Channels} 

In a number of papers (e.g, \cite{BW-Foils}), a {\em process} between probabilistic models $A$ and $B$ is defined to be a positive linear mapping $\phi : \V(A) \rightarrow \V(B)$ satisfying $u_{B}(\phi(\alpha)) \leq 1$ for 
all states $\alpha \in \Omega(A)$. This gives one a dual positive mapping $\phi^{\ast} : \V(B)^{\ast} \rightarrow \V(A)^{\ast}$ with $\phi^{\ast}(u_B) \leq u_A$. For our purposes, it's easier to work with these. So let's make it official, and also generalize a bit:

\begin{definition} 
If $\E$ and $\F$ are OUS-es, a (dual) {\em process} from $\E$ to $\F$ is a positive linear mapping $\Phi : \E \rightarrow \F$ with $\Phi(u_{\E}) \leq u_{\F}$. If this is an equality, we say that $\Phi$ is {\em unit-preserving,} or a {\em channel}.
\end{definition} 

Where $\E$ and $\F$ are 
the spaces of bounded linear operators on a pair of Hilbert spaces, this specializes to the usual notion of a quantum channel. Where $\E$ and $\F$ are the spaces of bounded measurable mappings on a pair of measurable spaces, this specializes to a Markov kernel (see Appendix B).

Any morphism of probabilistic models $\phi : A \rightarrow B$ gives rise to a process $\Phi : \E(A) \rightarrow \E(B)$ by setting $\Phi(\hat{a}) = \hat{\phi(a)}$. However, in general a process $\E(A) \rightarrow \E(B)$ needn't arise from such a morphism; indeed, it needn't pay any attention at all to the test-space structure of the model. However, as we'll see, it both gives rise to, and arises from, a PM-morphism $A^{o} \rightarrow B^{o}$. 

\begin{proposition}\label{prop: sub-ous} Let $\E$ be f.d..  If $v \in [0,u_{\E}]$, then the span $\E_{v}$ of $[0,v]$,  ordered by the cone $\E^{+}_{v} := \left \{ \sum_{i} t_i f_i \mid f_i \in [0,v], t_i \geq 0 \right \}$, 
is an OUS with $v$ as order unit. If $u$ is Archimedean, so is $v$ in $\E_{v}$. 
\end{proposition} 

The proof is routine, but in the interest of completeness, is given in Appendix A. Proposition \ref{prop: sub-ous} has several immediate consequences:

\tempout{
{\em Proof:} Since  $\F_{+} \subseteq \E_{+}$, we have $\F_{+} \cap -\F_{+} = \{0\}$, and $\F_{+}$ spans $\F$; thus, $\F_{+}$ is a proper cone, ordering $\F$. Now let $x \in \F$. Then for finitely many $f_i \in [0,v]$, 
$x = \sum_i  t_i f_i$ where $t_i \in \R$. Let $I = \{ i | t_i \geq 0\}$ and $J = \{ j | t_j < 0\}$. Then 
\[x = \underbrace{(\sum_{i \in I} t_i f_i )}_{a} - \underbrace{(\sum_{j \in J} |t_j| f_j)}_{b}.\]
Let $m = \sum_{i \in I} t_i$ and $n = \sum_{j \in J} |t_j|$. Then $0 \leq a \leq mu$ and $0 \leq b \leq nv$, whence, $-nv \leq x \leq mv$. It follows that $[-u,u]$ is absorbing, so $v$ is an order unit. $\Box$ 
}


\begin{corollary}\label{cor: procs to channels} Let 
$\Phi : \E \rightarrow \F$ be a process from an OUS $\E$ to 
an OUS $\F$. Then 
\vspace{-.2in}
\begin{itemize}  
\item[(a)] If $v = \Phi(u_{\E})$, then $\ran(\Phi) \subseteq \F_{v}$ and $\Phi : \E \rightarrow \F_{v}$ is a channel. 
\item[(b)] If $F$ is an $\E$-valued weight on $\M$, then $\Phi \circ F$ is an $\F_{v}$-valued weight on $\M$, where $v = F(u_{\E})$. 
\item[(c)] If $\Phi$ is a channel, then $\Phi$ extends to an $\F$-valued weight $(i,x) \mapsto \Phi(x)$ on $\M^{o}(\E)$.
\end{itemize}  
\end{corollary}



\begin{lemma}\label{lem: channels to morphisms}   
Any channel $\Phi : \E \rightarrow \F$ defines a test-preserving morphism $\phi : \M^{\J}(\E) \rightarrow \M^{o}(\F)$ by 
$\phi(x,a) = (x,\Phi(a))$. 
\end{lemma}

{\em Proof:} If $f : J \rightarrow \E$ is a $J$-valued observable on 
$\E$ for some $J \in \J$, then $\sum_{x \in J} \Phi(f(x)) = \Phi(u_{\E}) = u_{\F}$, so $\Phi \circ f$ is a $J$-valued observable on $\F$. 
Since $\phi$ is locally injective, it follows that it's a test-preserving test-space morphism. $\Box$

\begin{remark} Suppose now that $\D(\E)$ is the space of finite subsets of $[0,u]$ summing to $u$, regarded as a test space. Then $\D(\E)^{o} \subset \M^{o}(\E)$ is precisely the set of sequences $a_1,...,a_n$ of 
effects summing to $u$ {\em without repetition}. Any $\Phi$ as above gives us a mapping $\Phi : \D(\E) \rightarrow \D(\F)$. As noted earlier, is not a test-space morphism; but $\Phi^{o} : \M^{o}(\E) \rightarrow \M^{o}(\F)$ {\em is} a test-space morphism. \end{remark}

\section{Composites} 

It is widely understood nowadays that properties of composite systems --- in particular, the existence of entangled states and measurements --- represent the most important departure of QM from classical physics. The development of the convex-operational framework largely predated the emergence of this consensus, but there is no difficulty in bringing it up to date. 



\begin{definition} Let $\E$ and $\F$ be order unit space. 
A {\em composite} of $\E$ and $\F$ is an OUS $\G$, together 
with a bilinear mapping $\pi : \E \times \F \rightarrow \G$ 
such that  $\pi(a, b) \geq 0$ for all $a \in \E_+, b \in \F_+$, 
and $\pi(u_{\E}, u_{\F}) = u_{\G}$

Where we think of $\E$ and $\F$ as containing the effects 
associated with two physical systems, $\pi(a,b)$ is the 
``product effect" encoding the fact that $a$ occurred in a 
measurement on $\E$, and $b$ occurred in a measurement on $\F$. (Tacitly, we are assuming that the systems 
$\E$ and $\F$ can be jointly measured when brought together 
in the composite system $\G$.) 
\end{definition}  

Of course, the bilinear mapping $\pi$ associated with such a composite extends uniquely to a linear mapping, which we 
also denote by $\pi$, from $\E \otimes \F$ to $\G$. We say that the composite 
$(\G,\pi)$ is {\em locally tomographic} when this is an 
isomorphism.  In that case, one can identify $\G$ with 
$\E \otimes \F$, and treat $\G_{+}$ as 
a cone in the former. One finds that there are two extremal 
possibilities \cite{Wittstock}:  the {\em minimal} or {\em projective} cone $(\E \otimes_{\min} \F)_{+}$, consists of linear combinations with non-negative coefficients 
of pure tensors $a \otimes b$ with $a, b \geq 0$. The 
{\em maximal} cone, consisting of all elements $u \in \E \otimes \F$ such that $(f \otimes g)(u) \geq 0$ for 
all positive linear functionals $f$ and $g$ on $\E$ and $\F$, respectively. One then finds that 
\[(\E \otimes_{\min} \F)_{+} \subseteq \G_{+} \subseteq (\E \otimes_{\max} \F).\]

The minimal cone contains no entangled effects, but unless $\E$ or $\F$ is a simplex, its state space abounds with entangled states. Dually, the maximal cone contains a wealth of entangled effects, but supports no entangled states. The general non-signalling composite allows for both. 


Both the minimal and maximal tensor products are associative, and can be used to turn $\OUS$ into a symmetric monoidal category. More generally, Suppose $(\Cat, \Box)$ is a symmetric monoidal category, $T$ is an injective-on-objects functor $\Cat \rightarrow \OUS$, and $\pi$ $\pi$ is a natural transformation $T \mintensor T \rightarrow T \circ \otimes$. Thus, the components of $\pi$ are positive linear mappings 
\[\pi_{A,B} : T(A) \mintensor T(B) \rightarrow T(A \Box B)\]
When these make $(T(A \Box B), \pi_{A,B})$ a non-signalling 
composite for every $A, B \in \Cat$, we say that 
$(\Cat, \Box, \pi)$ is a (symmetric) {\em monoidal convex operational theory}.  

\begin{example} 
Standard non-relativistic QM can be regarded as an example. For $\k = {\mathbb C}$ or $\R$, let $\Cat$ be the category of $\k$-Hilbert spaces and isometries, and set $T(\H) = \B_{\sa}(\H)$, the space of bounded self-adjoint linear operators on $\H$, regarded as an order-unit space with the usual order and the identity operator as the order unit. If $f : \H \rightarrow \K$ is an isometry, then we have $T_{f} : \B_{\sa}(\H) \rightarrow \B_{\sa}(\K)$ given by $T_{f}(a) = f a f^{\ast}$. It is not hard to check that $T(\H \otimes \K) = \B_{\sa}(\H \otimes \K)$ is a non-signalling composite of $\B_{\sa}(\H)$ and $\B_{\sa}(\K)$ in the above sense. 
\end{example} 

There is a similar definition of non-signalling composites of probabilistic models, and of monoidal probabilistic theories \cite{Wilce-GPT}. Using this, one has 

\begin{theorem}\label{thm: monoidality} Let $(\Cat, \Box, \pi)$ be any monoidal convex operational theory, and let $\J$ be closed 
under finite Cartesian products. Then $\Mod^{~\J}$ is a monoidal functor into $\Prob$.
\end{theorem} 

The proof, along with relevant definitions, is sketched in Appendix C.

\tempout{ 
{\bf Observables} [Where to put this?] Let $\B$ be the Borel test space of a Boolean algebra $B$, and let $\phi : \B \rightarrow \M$ be an interpretation. Then (i) Can regard as an event-valued measure (explain how); (ii) If $\M$ is algebraic, $\phi$ becomes a measure $B \rightarrow \Pi(\M)$; (iii) regardless, $\phi(\B)$ is algebraic, and a refinement 
ideal in $\M$. (iv) Linearizing, we get a measure 
$B \rightarrow \E(\M)$ with range consisting of event-effects (in this sense, sharp). 

An observable $B \rightarrow \E$ gives $B \rightarrow \M^{\J}(\E)$? Well, $\tilde{F}(b) = ?$ ... OK, $B^{\J} \rightarrow \M^{\J}(\E)$, $\tilde{F}(i,b) = (i,F(b))$. Now, for any OA $B$, 
\[\M^{\J}(B) = \{ \{(i,b_i) | i = 1,...,n, \bigoplus b_i = 1\}\]

}

\begin{appendix} 

\section{Sub-order unit spaces} 

This Appendix contains the proof of Lemma \ref{prop: sub-ous}, which we restate in a 
slightly different form below. This is routine, and the result 
doubless well-known, but we have not found a reference for it. 

Let $(A,u)$ be an ordered vector space with order unit $u$.  Write 
$[0,u]$ for the interval of effects in $\A$.  For any set $X \subseteq A$, let $\spn_{+}(X) = \{ \sum_{i} t_{i} x_i | t_{i} \geq 0\}$.

\begin{definition} If $a \in [0,u]$, let 
\[A_{a} := \{ x \in A \mid \exists t > 0 ~ -ta \leq x \leq ta\}\]
and
\[A_{a}^{+} = \{ x \in A \mid \exists t > 0 ~ 0 \leq x \leq ta\}.\]
\end{definition} 

Note that $A_{a}^{+} \subseteq A_{a} \cap A_{+}$. 

\begin{lemma} $A_{a}^{+} = A_{a} \cap A_{+} = \spn_{+}([0,a]) = \{ x \in A_{+} | \exists t \geq 0 ~ x \leq ta$. 
\end{lemma} 

{\em Proof:} Clearly $A_{a}^{+} \subseteq A_{a} \cap A_{+}$. 
If $x \in A_{a} \cap A_{+}$, then  
$-ta \leq x \leq ta$ for some $t \geq 0$, and thus, $0 \leq x \leq ta$. 
We can choose $t > 0$, in which case 
$t^{-1}x \in [0,a]$, and thus $x \in \spn_{+}([0,a])$. Finally, 
if $x \in \spn_{+}([0,a])$, say $x = \sum_{i=1}^{n} t_i x_i$ 
with $0 \leq x_i \leq a$ and $t_i \geq 0$, then certainly $0 \leq x$. 
Let $t =  \max\{t_i\}$ and observe that $x \leq \sum_{i} t a = (nt)a$, 
so $x \in A_{a}^{+}$. 
$\Box$

\begin{lemma} $A_{a} = A_{a}^{+} - A_{a}^{+} = \spn([0,a])$. 
\end{lemma} 

{\em Proof:} By the above, $A_{a}^{+} = \spn_{+}([0,a])$.
Since we always have $\spn(X) = \spn_{+}(X) - \spn_{+}(X)$, the 
second equality follows, along with $\spn([0,a]) \subseteq A_{a}$. 
Now let $x \in A_{a}$, say with $-ta \leq x \leq ta$, with $t > 0$. Then 
$0 \leq x + ta \leq 2ta$, so $x + ta =: y \in A_{a}^{+}$, whence 
$x \in \spn_{+}([0,a])$, and thus $x = y + ta \in \spn([0,a])$. $\Box$

\begin{proposition} Let $a \in [0,u]$ as above. Then 
\vspace{-.15in}
\begin{itemize} 
\item[(a)] $(A_{a}, {A^{+}}_{a})$ is an ordered vector space; 
\item[(b)] For all $x, y \in A_{a}$, $x \leq y$ in $A_{a}$ iff 
$x \leq y$ in $A$; 
\item[(c)] $a$ is an order unit for $A_{a}$;
\item[(d)] If $u$ is Archimedean, so is $a$.
\end{itemize} 
\end{proposition} 

{\em Proof:} (a) Clearly, $A^{+}_{a} \cap -A^{+}_{a} \subseteq A_{+} \cap -A_{+} = \{0\}$, and Lemma 2 shows that $A^{+}_{a}$ spans $A_{a}$. 
So it remains to show that $A_{a}^{+}$ is a convex cone; but this 
is obvious. (b) If $x \leq y \in A_{a}$, then $y - x \in A_{a}^{+} 
\subseteq A_{+}$, so $x \leq y$ in $A$. Conversely, if $x \leq y$ in 
$A$, then $y - x \in A_{+} \cap A_{a} = A_{a}^{+}$ by Lemma 1, 
so $x \leq y$ in $A_{a}$. 
(c) That $a$ is an order unit is immediate from the definition 
of $A_{a}$. (d) Suppose now that $u$ is Archimedean: then 
if $x \in A_{a}$ with $nx \leq a$ in $A_{a}$ for all $n \in \N$, then 
by part (b), $nx \leq a$ in $A_{+}$, so as $a \leq u$, 
$nx \leq u$ in $A$, whence, $x \leq 0$. So $a$ is Archimedean 
in $A_{a}$. $\Box$

\section{Markov Kernels} 

A {\em Markov kernel} $k : S \rightarrow T$ from a measurable space $S = (S,\Sigma)$ to a measurable space $T = (T,\Xi)$ is a measurable function 
\[k : S \times \Xi \rightarrow [0,1]\]
such that 
\begin{itemize} 
\item[(a)] $k(~\cdot~, b)$ is a bounded measurable function 
on $S$ for every $b \in \Xi$, and 
\item[(b)] $k(s, ~\cdot~)$ is a countably-additive probability measure on $\Xi$ 
for every $s \in S$. 
\end{itemize}  

Let $B(S) = B(S,\Sigma)$ be the OUS consisting of bounded 
measurable real-valued functions on $S$, ordered pointwise, 
with order-unit $1$. For every $b \in \Xi$, 
let $\hat{k}(b) \in B(S)$ be the function defined by 
$\hat{k}(b)(s) = k(s,b)$ for all $s \in S$. 
Then (b) tells us that $k^{\ast}$ is a positive, countably-additive $B(S)$-valued measure on $\Xi$, with $k^{\ast}(T) = 1$. Note that countable additivity is with respect to the product topology 
on $B(S) \leq \R^{S}$, i.e., pointwise.  
For any measurable space $(S,\Sigma)$, there is a duality between $B(S,\Sigma)$ and the space $M_{\sigma}(S,\Sigma)$ 
of countably-additive finite measures on $(S,\Sigma)$, given by 
\[\langle \mu, f \rangle := \int_{S} f(s) \,d\mu(s),\]
under which $B(S) \simeq M_{\sigma}(S,\Sigma)^{\ast}$ \cite{Dunford-Schwartz}. Thus, for any measure $\mu \in M_{\sigma}(M,\Sigma)$, 
we have a countably additive measure $b \mapsto \langle \mu, \hat{k} (b) \rangle$ 
on $(T,\Xi)$. For any $g \in B(T)$ and $\mu \in M_{\sigma}(S,\Sigma)$ we can therefore define  
\[\hat{k}^{\ast}(g)(\mu) := \int_{T} g \,d\langle \mu, \hat{k} \rangle.\]
This gives us a positive linear mapping 
\[\hat{k}^{\ast} : B(T) \rightarrow B(S).\]
Then 
\begin{equation} 
\bar{k}^{\ast}(f)(s) = \langle \delta_{s}, \bar{k}^{\ast}(f)\rangle  = 
\int_{T} f d\delta_{x} \circ k^{\ast} 
= \int_{T} f(t) d k(s,~dt~).\end{equation}

Markov kernels $k : S \rightarrow T$ and $j : T \rightarrow U$ compose according to 
\[(j \circ k)(s,W) = \int_{T} j(t,W) d k(s,~\cdot~) .\]
This gives a category, $\Markov$. (Giorghiu and Heunnen 
call this {\bf SRel}, presumably for ``Stochastic relations".) But by (\theequation) above, this is simply
\[k^{\ast}( j^{\ast}(W)(\delta_{x})) = 
(k^{\ast} \circ j^{\ast})(W)(s).\]
That is, $(j \circ k)^{\ast} = k^{\ast} \circ j^{\ast}$. So 
we have a contravariant functor $\Markov \rightarrow \OUS$ 
taking each measurable space $S$ to $B(S)$, and 
each markov kernel $k : S \rightarrow T$ to the mapping $k^{\ast} : B(T) \rightarrow B(S)$. This gives us a 
convex-operational theory $\Markov^{\op} \rightarrow \OUS$. 
Composing with $\Mod^{\J}$, we obtain a probabilistic theory $\Markov^{\op} \rightarrow \Prob$. 

\section{Monoidality} 

In Section 5, we sketched the definition of a non-signalling composite of two order-unit spaces. This mirrors the definition of a non-signalling composite of probabilistic models given in \cite{Wilce-GPT}. We will briefly review this. If $\M$ and $\Ntest$ are test spaces, we write $\M \times \Ntest$ for the test space consisting of tests of the form $E \times F$, where $E \in \M$ and $F \in \Ntest$. The interpretation is that one can perform tests $E \in \M$ and $F \in \Ntest$ independently and collate the results as an ordered pair $(x,y)$ of outcomes $x \in E$ and $y \in F$. A probability weight on $\M \times \Ntest$ is {\em non-signalling} iff the marginal weights $\omega_{1}(x) = \sum_{y \in F} \omega(x,y)$ and $\omega_2(y) = \sum_{x \in E} \omega(x,y)$ are independent of $F$ and $E$, respectively.  Such a weight has well-defined {\em conditional} weights, defined by $\omega_{2|x}(y) = \omega(x,y)/\omega_1(x)$ and similarly for $\omega_{1|y}$. If $A$ and $B$ are probabilistic models, we say that a non-signalling probability weight $\omega$ $\M(A) \times \M(B)$ is a non-signalling {\em joint state} of $A$ and $B$  iff the conditional states $\omega_{1|y}$ and $\omega_{2|x}$ belong to $\Omega(A)$ and $\Omega(B)$, respectively (in which case, the marginal states do also). We define $A \times_{NS} B$ to be the model having test space $\M(A) \times \M(B)$ and state space consisting of all joint non-signalling states on $A$ and $B$. Then we make the following 

\begin{definition} A {\em non-signalling composite} of probabilistic models $A$ and $B$ is a pair $(C,\pi)$ where $\pi : A \times_{NS} B \rightarrow C$ is a test-preserving morphism of models. 
\end{definition} 

Thus, $\pi(x,y)$ is an event of $C$ representing the joint occurrence of $x$ and $y$ in measurements on $A$ and $B$, respectively, and $\pi^{\ast}(\gamma)(x,y) = \gamma(\pi(x,y))$ is a joint state of $A \times_{NS} B$. A {\em product state} on $\M(A) \times \M(B)$ is a joint 
state of the form $\alpha \otimes \beta$ where $\alpha \in \Omega(A)$ and $\beta \in \Omega(B)$ and $(\alpha \otimes \beta)(x,y) = \alpha(x) \beta(y)$. A convex combination of product states is said to be {\em separable}; a non-signaling state that is not separable is {\em entangled}. 

We say that $(C,\pi)$ is {\em strong} iff for every $\alpha \in \Omega(A)$ and $\beta \in \Omega(B)$, there exists some $\gamma \in \Omega(C)$ with $\pi^{\ast}(\gamma) = \alpha \otimes \beta$, and {\em locally tomographic} iff $\pi^{\ast}$ is injective.

One can show that $\V(A \times_{NS} B) \simeq \V(A) \maxtensor \V(B)$. If $\pi^{\ast}$ is injective. If $(C,\pi)$ is locally tomographic, we can regard $\V(C)_+$ as a cone in $\V(A) \otimes \V(B)$ contained in the the maximal tensor cone. $(C,\pi)$ is strong iff $\pi^{\ast}(\V(C)_+)$ contains the minimal tensor cone.

Suppose now that $\J$ is a collection of finite sets closed under the formation of Cartesian products,  and let $\E$, $\F$, $\G$ denote OUSes. If $\phi : \E \times \F \rightarrow G$ is a positive, unit-preserving bilinear form, we can define a positive, 
outcome- and test-preserving morphism 
 
$\tilde{\phi} : \M^{\J}(\E) \times \M^{\J}(\F) \rightarrow \M^{\J}(\G)$ by setting 
\begin{equation} 
\tilde{\phi}((x,a),(y,b)) = ((x,y),\phi(a,b))\end{equation} 
for any $x, y \in \bigcup \J$, and any effects $a \in E$ and $b \in F$. 
Note here that $(x,y) \in \bigcup \J$ by the standing assumption 
that $\J$ is closed under Cartesian products. 

\begin{proposition}\label{prop: monoidality} $(\M^{\J}(\G), \tilde{\phi})$ is a non-signalling composite of $\M^{H}(\E)$ and $\M^{\J}(\F)$. 
\end{proposition} 

{\em Proof:}  We first verify that $\phi$ is a bi-interpretation. 
Hence, if $\omega$ is a state on $G$, $\mu = \tilde{\phi}^{\ast}(\tilde{\omega})$ is probability weight on $\M^{\J}(\E) \times \M^{\J}(F)$. Note that 
\[\mu((x,a),(y,b)) = \tilde{\omega}((x,y),\phi(a,b)) = 
\omega(\phi(a,b))\]
by the definition of $\tilde{\omega}$. But this last is 
a bilinear form on $E \times F$. 

We now need to show that $\mu$ is non-signalling. If 
$f \in \M^{\J}(\E)$, say $f\in \Ob(I,\E)$ for some $I \in \J$, and 
if  $(y,b)$ is any outcome of $\M^{\J}(\F)$, say with 
$y \in J \in \J$ and $b = g(y)$ for some $g \in \Ob(J,\F)$, then 
by the above we have 
\[\sum_{(x,a) \in f} \mu((x,a),(y,b)) = \sum_{x \in I} 
\omega(\phi(f(x),b)) = \omega(\phi(\sum_{x \in I} f(x), b))
= \omega(\phi(u_E, b))\]
which is independent of $f$. Similarly, 
$\sum_{(y,b) \in g} \mu((a,x),(y,b))$ is independent of $g$. 

It remains to show that $\mu$'s conditional states 
lie in the state spaces of $E$ and $F$, respectively. For 
this, it suffices to show that 
\[(y,b) \mapsto \mu((x,a),(y,b)) = \omega(\phi(a,b))\]
belongs to the sub-normalized positive cone of $F^{\ast}$. 
But this only requires that it have the form 
$(y,b) \mapsto \psi(y,b)$ for some linear functional $\psi \in F^{\ast}$, which 
is clear from the linearity of $\omega$ and bilinearity 
of $\phi$. $\Box$ 

Suppose now that $(\Cat, \Box)$ is a symmetric monoidal category 
and $T : \Cat \rightarrow \OUS$ is a convex operational theory, 
i.e., an injective-on-objects functor from $F$ to $\OUS$ plus 
a natural transformation $\pi$ with bilinear components 
$\pi_{A,B} : F(A) \mintensor T(B) \rightarrow T(A \Box B)$ 
making $F(A \Box B)$ a nonsignalling composite of $T(A)$ and $T(B)$. 
For each $A, B \in \Cat$, let $\pi^{b}_{A,B}$ be the positive 
bilinear mapping $T(A) \times T(B) \rightarrow T(A \Box B)$ 
obtained by restricting $\pi$ to pure tensors. 
Then we have a functor 
$F^{\J} = \Mod^{\J} \circ T : \Cat \rightarrow \Prob$ 
and a natural transformation $\pi^{\J}$ with components 
\[\pi^{\J}_{A,B} = \tilde{\pi^{b}_{A,B}} : \Mod^{\J}(T(A)) \times_{NS} \Mod^{\J}(T(B)) \rightarrow \Mod^{\J}(T(A \Box B))\]
as defined in Equation (\theequation). 
By Proposition \ref{prop: monoidality}, this makes 
$\Mod^{\J}(T(A \Box B))$ a nonsignalling composite of 
$\Mod^{\J}(T(A))$ and $\Mod^{\J}(T(B))$. Thus, we have


\begin{theorem}[Theorem \ref{thm: monoidality}, {\em bis}] 
Let $\J$ be closed under finite Cartesian products. If $T : \Cat \rightarrow \OUS$ is any monoidal convex operational theory,  then $\Mod^{~\J} \circ T : \Cat \rightarrow \Prob$ is a monoidal probabilistic theory.
\end{theorem}

\section{Rolling Dice} 

Here is another approach to understanding arbitrary 
discrete observables in operational terms. The idea is simple:
if $A \subseteq E \in \M$ is an event, $a \in \E$ is an effect, and for every $x \in A$, 
$F(x) = p_x a$ for some $p_x > 0$, then we can simulate $F$'s behavior on $A$ by 
performing the test $E$, recording $A$ if an outcome in $A$ occurs, and then ``rolling a die'' 
--- specifically, an $|A|$-sided die, with faces labeled by the outcomes of $A$, weighted so that side ``$x$'' turns up with probability $p_x$. 

{\bf Coarsenings, Dacey sums, and graphs} 
To formalize this, it will be helpful to recall two basic constructions with test spaces. First, if $\M$ is any test space, $\M^{\#}$ denotes its {\em coarse-graining}: the test space consisting of finite partitions of tests $E \in \M$ by non-empty events. Thus, 
every non-empty event of $\M$ is an {\em outcome} of $\M^{\#}$. Probability weights on $\M$ lift to probability weights on $\M^{\#}$ in an obvious way. We also have a canonical morphism $\M^{\#} \rightarrow \M$, taking event $A$ {\em qua} outcome to $A$, {\em qua} event. (Strictly speaking, this is just the inclusion mapping from the set of non-empty events to that of all events of $\M$.) 

Secondly, if $\M$ is a test space with outcome-set $X$, 
and $\{\Ntest_{x}\}$ is any $X$-indexed collection of test spaces, the {\em Dacey sum} (or {\em fibred product}) of 
the latter over the former is the test space $\Dacey^{\M}_{x \in X} \Ntest_{x}$ consisting of two-stage tests of the form 
\[\bigcup_{x \in E} \{x\} \times F_{x}\]
where $E \in \M$ and, for each $x \in E$, $F_{x} \in \Ntest_{x}$. Operationally, to perform such a test, one 
performs the initial test $E$; upon obtaining outcome $x$, 
one then performs the test $F_{x} \in \Ntest_x$ and, if 
this yields outcome $y$, records the pair $(x,y)$ as the 
outcome of the entire two-stage test. 

An important special case of this is the {\em graph} of a 
morphism $\phi : \M \rightarrow \Ntest$. This is the test-space  $\Graph(\phi)$ consisting of sets of the form 
\[\bigcup_{x \in E} \{x\} \times \phi(x)\]
where $E \in \M$. Note that here we consider the {\em event} 
$\phi(x)$ as the outcome-set of a classical test space. That 
is, $\Graph(\phi)$ is the Dacey sum $\Dacey^{\M}_{x \in X} \{\phi(x)\}$. 

{\bf The Dacey cover} We can apply the graph construction to the canonical morphism $\phi : \M^{\#} \rightarrow \M$. That is, for every event $A \in \mathcal{E}(\M)$, we can regard $A$ as the outcome-set of  the classical test space $\{A\}$ --- the ``$|A|$-sided die" mentioned 
earlier --- and form the Dacey sum 
\[\Dacey(\M) := \Graph(\phi) = \Dacey^{\M^{\#}}_{A \in \bigcup \M^{\#}} \{A\} \times A.\]
An outcome of $\Dacey(\M)$, then, is a pair $(A,x)$ where $A$ is a non-empty event and $x \in A$. A test in $\Dacey(\M)$ is a set $\{(A_i,x_i)\}$ of such pairs such that $\{A_i\}$ partitions a test $E \in \M$. Operationally: perform the coarse-grained test $\{A_i\}$; upon obtaining the {\em outcome} $A_i$, perform the classical test $A_i$ to obtain an outcome in $E$. We will refer to $\Dacey(\M)$ as the {\em Dacey cover} of $\M$. 


If we like, we can understand the classical test $A_i$ as simply ``remembering" which outcome of $E = \bigcup A_i$ actually occurred, 
but we prefer here to imagine that, when we perform 
the coarse-grained test $\{A_i\}$, all information about the 
``actual" outcome is lost: we only seem some indication that 
one of the events $A_i$ occurred. In this case, we literally 
follow this indication with a second experiment, in which we 
(somehow) obtain an outcome in $A_i$ --- say, by rolling the 
aforesaid die. To emphasize this point of view, we'll write 
$\D_{A}$ for the classcal test space $\{A\}$. 

There are canonical morphisms $\pi : \Dacey(\M) \rightarrow \M$ 
and $\psi : \M^{\#} \rightarrow \Dacey(\M)$ given by $\pi(A,x) = x$ 
and $\psi(A) = \{(A,x) | x \in A\}$, respectively.  These are both test-preserving; $\psi$ is injective, but not outcome-preserving, 
while $\pi$ is outcome-preserving, but not injective. 
If $\phi : \M^{\#} \rightarrow \M$ is the canonical 
morphism described above, then $\phi = \pi \circ \psi$ 
\[
\begin{tikzcd} 
\M^{\#} \arrow[dd, "\phi"]
 \arrow[rr, "\psi"] 
& & \Dacey(\M) \arrow[ddll,  "\pi"]\\
& & \\
\M & & 
\end{tikzcd} 
\]
It is straightforward that a probability weight $\beta$ on $\Dacey(A)$ has the form 
\[\beta(A,x) = \alpha(A) \beta_{A}(x)\]
where $\beta_{A}(x)$ is simply a chosen probability weight on the ``die'' ${\mathcal D}_{A}$. 

{\bf De-randomizing an $\E$-valued weight} 
Let $\E$ be an OUS, and suppose $F$ is an $\E$-valued weight  $\M$.  This extends to an $\E$-valued weight $F^{\#} : \M^{\#} \rightarrow \E$ by setting 
$F^{\#}(A) = \sum_{x \in A} F(x)$ for all non-empty events $A$. In general, we won't distinguish between $F$ and $F^{\#}$, i.e., we adopt the convention that for an $\E$-valued weight  on $\M$, and an event $A$ of $\M$, $F(A)$ always {\em means} $\sum_{x \in A} F(x)$ (rather than, say, $\{ F(x) | x \in A\}$.)  Also, if $\M$ and $\Ntest$ are test spaces, $\psi : \M \rightarrow \Ntest$ is a morphism, and $F : \Ntest \rightarrow \E$ is an 
$\E$-valued weight on $\Ntest$, we define a state $F \circ \psi$ on $\M$ by setting 
\[(F \circ \psi)(x) = F(\psi(x)) = \sum_{y \in \psi(x)} F(y).\]

Let $\approx$ be the equivalence relation on $[0,u]$ defined by $a \approx b$ iff for some $t > 0$, $a = tb$. This pulls back along $F$ to an equivalence relation on $X = \bigcup \M$, namely $x \approx y$ iff $F(x) \approx F(y)$ for $x,y \in X$.  

\begin{definition} 
Let $R \subseteq X/\approx \times \M$ be the set of all pairs $(C,E)$, where $C \subseteq X$ is an equivalence class under $\approx$ and $E$ is a test, such that $C \cap E \not = \emptyset$. For each $E \in \M$, we 
now have a partition (a coarse-graining) defined by 
\[\hat{E} := \{ C \cap E | (C,E) \in R\}.\] 
Let $\hat{\M} \subseteq \M^{\#}$ be the collection of all such partitions $\hat{E}$ 
with $E \in \M$.    
\end{definition} 

The construction of $\hat{\M}$ depends on the $\E$-valued weight $F$, 
so whenever there's a possibility of ambiguity, we write $\hat{\M}_{F}$.  

Let $(C,E) \in R$, and let $x_o \in C \cap E$. Then for each $x \in C \cap E$, there exists a unique $t(x) > 0$ with $F(x) = t(x) F(x_o)$ (in particular, $t(x_o) = 1$). Note that 
this depends only on $x$ and $x_o$. Now 
\[F(C \cap E) = \sum_{x \in C \cap E} F(x) = \left (\sum_{x \in C \cap E} t(x)\right )F(x_o).\]
Setting $a_{C,E} = F(C \cap E) \in \E$ and $T_{C,E} = \sum_{x \in C \cap E} t(x) > 0$, we have $a_{C,E} = T_{C,E} F(x_o)$, so 
\[F(x_o) = \frac{1}{T_{C,E}} a_{C,E}\]
and thus, for all $x \in C \cap E$, 
\[F(x) = \frac{t(x)}{T_{C,E}} a_{C,E}.\]
Notice that this shows that the fraction $t(x)/T_{C,E}$ is independent of the choice of 
$x_o$. 

{\em Notation:} From now on, for any equivalence class $C$, test $E$, and outcome $x \in C \cap E$, we write $p_{C,E}(x)$ for $t(x)/T_{C,E}$, noting that this defines a probability 
weight, in the classical sense, on the die $\D_{C \cap E}$.  

We now have a natural $\E$-valued weight on 
$\hat{\M}$, namely 
\[\hat{F}(C \cap A) := F(C \cap A) = \sum_{x \in C \cap E} F(x) =: a_{C,E},\]
 This is locally injective, and thus, a morphism $\hat{\M} \rightarrow \Dacey(\E)$. 

\begin{definition} Let $\Dacey_{F}(\M)$ be the subset of $\Dacey(\M)$ 
consisting of two-stage tests with initial test in $\hat{\M}$. 
Equivalently:  $\Dacey_{F}(\M)$ is the graph of the restriction of $\phi$ to $\hat{\M}$. 
\end{definition} 

 Given a state $\alpha \in S(\E)$, we can define a probability weight $\tilde{\alpha}$ on 
 $\hat{\M}$ by setting $\hat{\alpha} = \hat{F}^{\ast}(\alpha)$; that is, if  $C \cap E \in \bigcup \hat{\M}$, we have 
 \[\hat{\alpha}(C \cap E) = \alpha(a_{C,E}).\]
 We now define a further probability weight $\tilde{\alpha}$ on $\Dacey_{F}(\M)$: 
 \[\tilde{\alpha}(C \cap E, x) = \hat{\alpha}(C \cap E) p_{C,E}(x).\]
 Note that the second-stage probability weight $p_{C,E}$ depends only on $C$ and $E$,  and not at all on the state $\alpha$.  Moreover, we have 
 \[\alpha(x) = \tilde{\alpha}(C \cap E, x),\]
 so we see that the probability of obtaining $x$ in an execution of $E$ depends on 
 the probability of obtaining $C \cap E$ --- which does depend on $\alpha$ --- followed by the chance of selecting $x \in C \cap E$ through a randomization process that has nothing to do with $\alpha$.  

\tempout{  
There is a natural morphism $\pi : \Dacey_{F}(\M) \rightarrow \M$, given by 
\[\pi((C \cap E), x) = x\]
and an $\E$-valued weight $\tilde{F}$ on $\Dacey_{F}(\M)$, given by $\tilde{F}((C \cap E), x) = F(x)$. 
This gives us the following diagram: 
\[
\begin{tikzcd} 
\hat{\M} 
\arrow[rr, shift right=.7ex,"\psi"] 
\arrow[dd, shift left=.75ex, "\phi"]   
\arrow[ddrr, dashed, "\hat{F}", near start]
& & 
\Dacey_{F}(\M)
 \arrow[dd, "\tilde{F}"] \arrow[ddll, "\hspace{-.07in} \pi", near start]\\
& & \\
\M 
\arrow[rr, "F"] 
& &  \E
\end{tikzcd}
\]
The bottom-left triangle commutes by definition of $\hat{F}$, and the top-left and bottom-right triangles commute by definition of $\hat{F}$ and $\tilde{F}$, respectively. Hence, do the square and the top-right triangle also commutes. 
 
The range of the morphism $\hat{F}$ lies in a smaller test space than $\Dacey(\F)$: 

\begin{definition} 
Let $\Dacey_{o}(\E)$ denote the sub-test space of $\Dacey(\E)$ consisting of unordered partitions of unity containing no distinct but equivalent elements. 
\end{definition} 

The upshot, then: Any $\E$-valued weight on $\M$ can be explained by a morphism 
$\hat{F} : \hat{\M} \rightarrow \Dacey_o(\E)$, followed by state-independent classical randomization.  A perhaps more informative diagram is the following: 
\[
\begin{tikzcd} 
\hat{\M} 
\arrow[rr, shift right=.7ex,"\psi"] 
\arrow[dd, shift left=.75ex, "\hat{F}"]   
& & 
\Dacey_{F}(\M)
 \arrow[dd, "\tilde{F}"] \arrow[dr, "\pi"]
 & \\ 
& & & \M \arrow[dl, "F"]\\
\Dacey_o(\E)
\arrow[rr, "G"] 
& &  \E &  
\end{tikzcd}
\]
where $G$ is the obvious $\E$-valued weight on $\D_o(\E)$.  We can think of $\hat{F}$ as {\em de-randomizing} $F$; passing to $\hat{F}$ adds back the necessary randomization to recover $F$. 
}

\end{appendix}

\end{document}